\documentclass[sigconf]{acmart}
\usepackage{booktabs}
\usepackage{bm}
\usepackage{comment}
\usepackage{soul}
\usepackage{caption}
\usepackage{float}
\usepackage{amsfonts}
\usepackage[normalem]{ulem}
\usepackage[]{algorithm2e}
\usepackage{todonotes}
\usepackage{multirow, makecell}
\usepackage{hyperref}
\usepackage{subcaption,siunitx,booktabs}
\usepackage{graphicx}

\copyrightyear{2021}
\acmYear{2021}
\setcopyright{iw3c2w3}
\acmConference[WWW '21]{Proceedings of the Web Conference 2021}{April 19--23, 2021}{Ljubljana, Slovenia}
\acmBooktitle{Proceedings of the Web Conference 2021 (WWW '21), April 19--23, 2021, Ljubljana, Slovenia}

\acmPrice{}
\acmDOI{10.1145/3442381.3449930}
\acmISBN{978-1-4503-8312-7/21/04}

\begin{document}

\title{$FM^2$: Field-matrixed Factorization Machines \\for Recommender Systems}


\author{Yang Sun}
\affiliation{%
  \institution{Yahoo Research}
  \streetaddress{701 First Ave.}
  \city{Sunnyvale}
  \state{CA}
  \country{USA}}
\email{yang.sun@verizonmedia.com}

\author{Junwei Pan}
\affiliation{%
  \institution{Yahoo Research}
  \streetaddress{701 First Ave.}
  \city{Sunnyvale}
  \state{CA}
  \country{USA}}
\email{pandevirus@gmail.com}

\author{Alex Zhang}
\affiliation{%
  \institution{Yahoo Research}
  \streetaddress{701 First Ave.}
  \city{Sunnyvale}
  \state{CA}
  \country{USA}}
\email{alex.zhang@verizonmedia.com}

\author{Aaron Flores}
\affiliation{%
  \institution{Yahoo Research}
  \streetaddress{701 First Ave.}
  \city{Sunnyvale}
  \state{CA}
  \country{USA}}
\email{aaron.flores@verizonmedia.com}



\begin{abstract}

Click-through rate (CTR) prediction plays a critical role in recommender systems and online advertising. The data used in these applications are multi-field categorical data, where each feature belongs to one field. Field information is proved to be important and there are several works considering fields in their models. In this paper, we proposed a novel approach to model the field information effectively and efficiently. The proposed approach is a direct improvement of FwFM, and is named as Field-matrixed Factorization Machines (FmFM, or $FM^2$). We also proposed a new explanation of FM and FwFM within the FmFM framework, and compared it with the FFM. Besides pruning the cross terms, our model supports field-specific variable dimensions of embedding vectors, which acts as a soft pruning. We also proposed an efficient way to minimize the dimension while keeping the model performance. The FmFM model can also be optimized further by caching the intermediate vectors, and it only takes thousands floating-point operations (FLOPs) to make a prediction. Our experiment results show that it can out-perform the FFM, which is more complex. The FmFM model's performance is also comparable to  DNN models which require much more FLOPs in runtime.

\end{abstract}



\maketitle

\section{Introduction}
Click-through rate (CTR) prediction plays a key role in recommender systems and online advertising, and it has attracted much research attention in the past decade~\cite{chapelle2015simple, mcmahan2013ad, richardson2007predicting, yang2011optimal, deng2017pricing, sun2012convergence}. The data involved in CTR prediction are typically \emph{multi-field categorical data}~\cite{zhang2016deep,pan2018fwfm}. Such data possess the following properties. First, all the features are categorical and are very sparse since many of them are identifiers. Therefore, the total number of features can easily reach millions to billions. Second, every feature belongs to one and only one field and there can be tens to hundreds of fields. 

A prominent model for these prediction problems is logistic regression with cross-features~\cite{chapelle2015simple}. When all cross-features are considered, the resulting model is equivalent to a polynomial kernel of degree 2~\cite{chang2010training}. However, it takes too many parameters to consider all possible cross-features. To resolve this issue, matrix factorization~\cite{koren2009matrix,aharon2013off} and factorization machines (FM)~\cite{rendle2010factorization,rendle2012factorization} was proposed to learn the effects of cross features by dot products of two feature embedding vectors. Based on FM, Field-aware Factorization Machines (FFM)~\cite{juan2016field,juan2017field} was proposed to consider the \emph{field} information to model the different interaction effects of features from different field pairs. Recently, a Field-weighted Factorization Machine (FwFM)~\cite{pan2018fwfm,pan2019predicting} model was proposed to consider the field information in a more parameter-efficient way.

Existing models that consider the field information either has too many parameters, such as FFM~\cite{juan2016field,juan2017field}, or is not very effective, such as~\cite{pan2018fwfm}.  
We propose to use a field matrix between two feature vectors to model their interactions, where the matrix is learned separately for each field pair. 
We will show that our field-pair matrix approach achieves good accuracy performance while maintaining computational space and time efficiency.

\section{Related Works Overview}
\label{sec:related_works}
\textbf{Logistic Regression (LR)} is the most widely used model on multi-field categorical data for CTR prediction~\cite{chapelle2015simple,richardson2007predicting}. Suppose there are $m$ unique features $\{f_1,\cdots,f_{m}\}$ and $n$ different fields $\{F_1,\cdots,F_{n}\}$. Each field may contain multiple features, while each feature belongs to only one field. To simplify the  notation, we use index $i$ to represent feature $f_i$, and $F(i)$ to represent the field which $f_i$ belongs to. Given a data set $\bm{S} = \{y^{(s)},\bm{x}^{(s)}\}$, where $y^{(s)}\in\{1,-1\}$ is the label (clicked or not) and $\bm{x}^{(s)} \in \{0,1\}^m$ is the feature vector in which $x_i^{(s)} = 1$ if feature $i$ is active for this instance otherwise $x_i^{(s)}=0$, the LR model parameters $\bm{w}$ are estimated by minimizing the following loss function:
\begin{equation}
	\min_{\bm{w}}\;[\sum_{s=1}^{|S|}\log(1+\exp(-y^{(s)} \Phi_{LR}(\bm{w}, \bm{x}^{(s)}))) + \lambda \|\bm{w}\|_2^2 ]
\end{equation}
The first term is the log loss, and the second term is the L2 regularization term where $\lambda$ is the regularization weight, and
\begin{equation}
\Phi_{LR}(\bm{w}, \bm{x}) = w_0 + \sum_{i=1}^m x_i w_i
\end{equation}
is a linear combination of individual features.

However, linear models lack the capability to represent the feature interactions \cite{chapelle2015simple}. As cross features may have more important factors than those single features, many improvements have been proposed in the past decades.

\textbf{Degree-2 Polynomial  (Poly2)} models as a general way to address this problem is to add feature conjunctions. It has been shown that Poly2 models can effectively capture the effect of feature interactions\cite{chang2010training}. Mathematically, in the loss function of equation (1), Poly2 models consider replacing $\Phi_{LR}$ with
\begin{equation}
\Phi_{Poly2}(\bm{w}, \bm{x}) = w_0 + \sum_{i=1}^m x_i w_i + \sum_{i=1}^{m}\sum_{j=i+1}^m x_{i}x_{j}w_{h(i,j)}
\end{equation}
where $h(i,j)$ is a function which hashes feature conjunction $(i,j)$ into a natural number in the hashing space $H$ to reduce the number of parameters. Otherwise the number of parameters in the model would be in the order of $O(m^2)$, which is too many to be learned.

\textbf{Factorization Machines(FM)} learn an embedding vector $\bm{v}_i \in \mathbb{R}^K$ for each feature, where $K$ is a hyper-parameter and is usually a small integer, e.g., 10. FM model the interaction between two features $i$ and $j$ as the dot product of their corresponding embedding vectors $\bm{v}_i$, $\bm{v}_j$:

\begin{equation}
\Phi_{FM}\left((\bm{w},\bm{v}),\bm{x} \right)=  w_0 + \sum_{i=1}^m x_i w_i +\sum_{i=1}^{m}\sum_{j=i+1}^m x_{i}x_{j} \langle \bm{v}_{i}, \bm{v}_{j}\rangle
\end{equation}

FM usually outperform Poly2 models in applications involving sparse data such as CTR prediction. This is because it models the interaction between two features by a dot product between their corresponding embedding vectors. These embedding vector of a feature is meaningful as long as the this feature appears enough times during model training. However, FM neglect the fact that a feature might behave differently when it interacts with features from different other fields.

\textbf{Field-aware Factorization Machines (FFM)} model such difference explicitly by learning $n-1$ embedding vectors for each feature, say $i$, and only using the corresponding one $\bm{v}_{i, F(j)}$ to interact with another feature $j$ from field $F(j)$:
\begin{equation}
	\Phi_{FFM}((\bm{w}, \bm{v}), \bm{x})= w_0 + \sum_{i=1}^m x_i w_i + \sum_{i=1}^{m}\sum_{j=i+1}^m x_{i}x_{j} \langle \bm{v}_{i, F(j)}, \bm{v}_{j, F(i)}\rangle 
\end{equation}

Although FFM have gotten significant performance improvements over FM, their number of parameters is in the order of $O(mnK)$. The huge number of parameters in FFM is undesirable in the real-world production systems~\cite{juan2017field}. Therefore, it is appealing to design alternative approaches that are competitive and more memory-efficient.

\textbf{Field-weighted Factorization Machines (FwFM)} was proposed in ~\cite{pan2018fwfm}, which models the different field interaction strength explicitly. More specifically, the interaction of a feature pair $i$ and $j$ in our proposed approach is modeled as

\[x_ix_j\langle \bm{v}_i, \bm{v}_j\rangle r_{F(i), F(j)}\]
\label{def:fwfm}

\noindent where  $\bm{v}_i, \bm{v}_j$ are the embedding vectors of $i$ and $j$, $F(i), F(j)$ are the fields of features $i$ and $j$, respectively, and $r_{F(i), F(j)} \in \mathbb{R}$ is a weight to model the interaction strength between fields $F(i)$ and $F(j)$. The formulation of FwFM is:

\begin{equation}\label{eq:fwfm}
	\Phi_{FwFM}((\bm{w},\bm{v}), \bm{x}) = w_0 + \sum_{i=1}^m x_i w_i +  \sum_{i=1}^m\sum_{j=i+1}^m x_{i} x_{j} \langle \bm{v}_i, \bm{v}_j \rangle r_{F(i), F(j)}
\end{equation}

FwFM are extensions of FM in the sense that it uses additional weight $r_{F(i), F(j)}$ to explicitly capture different interaction strengths of different field pairs. FFM can model this implicitly since they learn several embedding vectors for each feature $i$, each one $\bm{v}_{i, F_k}$ corresponds to one of other fields $F_k \neq F(i)$, to model its different interactions with features from different fields. However, the model complexity of FFM is significantly higher than that of FM and FwFM.

Recently, there are also lots of work on deep learning based click prediction models~\cite{cheng2016wide,zhang2016deep,qu2016product,guo2017deepfm,shan2016deep,he2017neural,wang2017deep,zhou2018deep,xdeepfm,autoint}. These models capture both low order and high order interactions and achieve significant performance improvement. However, the online inference complexity of these models is much higher than the shallow models~\cite{deng2020sparse}. Model compression techniques such as pruning~\cite{deng2020sparse}, distillation~\cite{zhou2018rocket} or quantization are usually needed to accelerate these models in the online inference. In this paper, we focus on improving the low order interactions, and the proposed model can be easily used as a shallow component in these deep learning models.

\section{Our Model}
We propose a new model to represent the interaction of field pairs as a matrix. Similar to FM and FwFM, we learn an embedding vector for each feature. We define a matrix $M_{F(i),F(j)}$ to represent the interaction between field $F(i)$ and field $F(j)$

\[x_ix_j\langle \bm{v}_i  M_{F(i), F(j)}, \bm{v}_j\rangle\]

\noindent where  $\bm{v}_i, \bm{v}_j$ are the embedding vectors of feature $i$ and $j$, $F(i), F(j)$ are the fields of feature $i$ and $j$, respectively, and $M_{F(i), F(j)} \in \mathbb{R}^{K\times K}$ is a matrix to model the interaction between field $F(i)$ and field $F(j)$. We name this model Field-matrixed Factorization Machines (FmFM):
\begin{equation}\label{eq:fmfm}
	\Phi_{FmFM}((\bm{w},\bm{v}), \bm{x}) = w_0 + \sum_{i=1}^m x_i w_i +  \sum_{i=1}^m\sum_{j=i+1}^m x_ix_j\langle \bm{v}_i  M_{F(i), F(j)}, \bm{v}_j\rangle\
\end{equation}

FmFM are extensions of FwFM in that it uses a 2-dimensional matrix $M_{F(i), F(j)}$  to interact different field pairs, instead of a scalar weight $r$ in FwFM. With those matrices, features from the embedding space can be transferred to $n-1$ spaces; we name those matrices Field-matrices. Figure \ref{fig:fmfm} demonstrates the calculation of the interaction pairs $(v_i, v_j)$ and $(v_i, v_k)$, while features $i, j$ and $k$ are from 3 different fields.

\begin{figure}[h]
  \centering
  \includegraphics[width=0.9\linewidth]{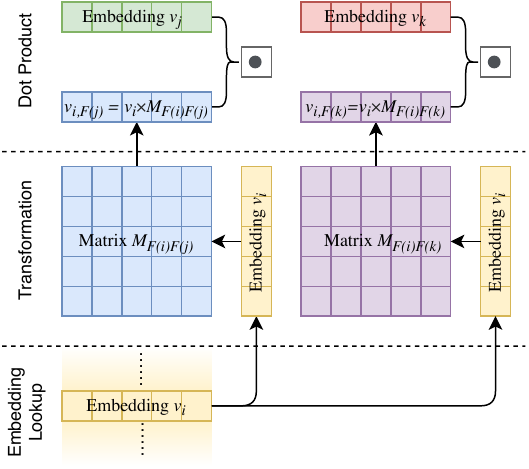}
  \caption{An example of FmFM interaction terms calculation}
  \label{fig:fmfm}
\end{figure}

The calculation can be decomposed into 3 steps:
\begin{enumerate}
\item \textbf{Embedding Lookup:} The feature embedding vectors $v_i, v_j$, and $v_k$ are looked up from the embedding table, and $v_i$ will be shared between those 2 pairs.

\item \textbf{Transformation:} Then $v_i$ is multiplied by the matrices $M_{F(i)F(j)}$ and $M_{F(i)F(k)}$ respectively, here we get the intermediate vector $v_{i,F(j)}=v_i\times M_{F(i)F(j)}$ for the field $F(j)$, and $v_{i,F(k)}=v_i\times M_{F(i)F(k)}$ for the field $F(k)$.

\item \textbf{Dot Product:} The final interaction terms will be a simple dot product between $v_j$ and $v_{i,F(j)}$, as well as  $v_k$ and $v_{i,F(k)}$, which are the black dots showed in Fig.\ref{fig:fmfm}.
\end{enumerate}

\subsection{The United Framework of Factorization Machines' Family}

FmFM have a similar design with, while extending, FM and FwFM; in this section, we deep dive into their design, explain their structure with the 3-step FmFM framework above, and figure out the fundamental relationships among these factorization machine models.

\subsubsection{FM}
Figure \ref{fig:fm} shows the  calculation of feature interactions in FM, the difference to FmFM is that FM skip the step 2, and use the shared embedding $v_i$ to do the final dot product with $v_j$ and $v_k$ respectively. Since we know
\[v_i=v_i I_K, \]
we can construct an identity matrix $I_K$ and let all field matrices equal to $I_K$. As the identity matrix shows in Fig.\ref{fig:fm}, the FM actually is a special case of FmFM when all field matrices are $I_K$. Since those matrices $I_K$ are fixed and non-trainable, we define its degree of freedom to be 0.

\begin{figure}[h]
  \centering
  \includegraphics[width=0.8\linewidth]{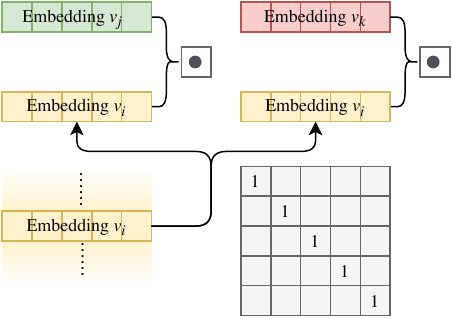}
  \caption{An explanation of FM with FmFM framework}
  \label{fig:fm}
\end{figure}

\subsubsection{FwFM}
Fig.\ref{fig:fwfm} shows the calculation of feature interactions in FwFM. There is a change from the original definition \ref{def:fwfm}, while, it is easy to know that:
\[\langle \bm{v}_i, \bm{v}_j\rangle r_{F(i), F(j)} = \langle \bm{v}_i r_{F(i), F(j)}, \bm{v}_j\rangle \]
Thus, we calculate the term $\bm{v}_i r_{F(i), F(j)}$ firstly in figure \ref{fig:fwfm}, instead of $\langle \bm{v}_i, \bm{v}_j\rangle$ in the original definition in Eq.\ref{def:fwfm}. It is clear now that the intermediate vector in step 2 is actually a scaled embedding vector:
\[v_{i,F(j)} = v_i r_{F(i)F(j)} = v_i (r_{F(i)F(j)} I_K)\]

Thus, we construct the field matrix in FwFM as a scalar matrix $r_{F(i)F(j)} I_K$, which is a diagonal matrix with all its main diagonal entries equal $r$. Its effect on the embedding vector $v_i$ is a scalar multiplication by $r$. We show this matrix at the corner of Fig.\ref{fig:fwfm} (left one). Since the scalar $r$ is trainable, it has one more degree of freedom than FM, we define its degree of freedom as 1.

\begin{figure}[h]
  \centering
  \includegraphics[width=0.9\linewidth]{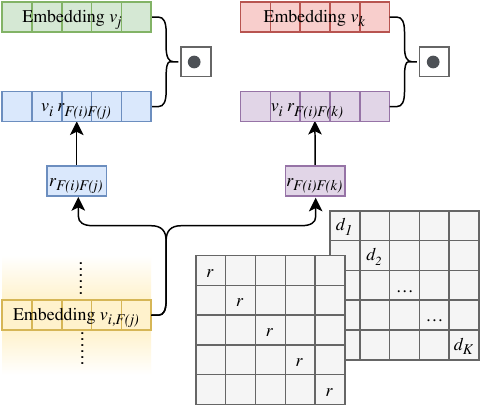}
  \caption{An explanation of FwFM and FvFM with FmFM framework}
  \label{fig:fwfm}
\end{figure}

\subsubsection{FvFM}
We follow the clue above, and give one more freedom to the field matrix in FwFM. Let the field matrix become a  diagonal matrix in which the main diagonal entries are trainable variables, instead of one single variable in FwFM, we can rewrite the intermediate vector:
\[v_{i,F(j)} = v_i D_{F(i)F(j)} = v_i \odot d_{F(i)F(j)},\]

where $D_{F(i)F(j)}=\mathrm{diag}(d_1, d_2,\dots, d_K)$,  this can be expressed more compactly by using a vector $d_{F(i)F(j)}\in \mathbb{R}^K$ instead of the diagonal matrix, and taking the Hadamard product ($\odot$) of the vectors $v_i$. Figure \ref{fig:fwfm} demonstrates this case in the right matrix at the corner.

We name this method \textbf{Field-vectorized Factorization Machines (FvFM)}. The FvFM have one more freedom than FwFM: the trainable parameters become a vector instead of a scalar; thus, we define its degree of freedom to be 2.

\subsubsection{FmFM}
Let's revisit FmFM in figure \ref{fig:fmfm}. It has all the degrees of freedom of a matrix, which is 3. 
All the variables in those matrices are trainable, and we expect the FmFM to have a greater predictive capacity than other factorization machine models. We will evaluate this hypothesis in the next section.

Overall, we have found that FM, FwFM, FvFM are all special cases of FmFM, the only differences are their field matrices' restrictions. According to their flexibility, we summarize them in the Table\ref{table:model_df}

\begin{table}[h]
\centering
	\begin{tabular}{| c | c | c |}
    \hline
    Model & Field Interaction & Degree of Freedom \\
    \hline
    FM & Constant & 0 \\
    \hline
    FwFM & Scalar & 1\\
    \hline
    FvFM & Vector & 2 \\
    \hline
    FmFM & Matrix & 3\\
    \hline
	\end{tabular}
    \caption{Degrees of freedom in different FM models}
    \label{table:model_df}
\end{table}

\subsubsection{Connections to OPNN}

FmFM can also be viewed as modeling the interaction of two feature embedding vectors by a weighted outer product: 

\begin{equation}\label{eq:fmfm}
	\Phi_{FmFM}((\bm{w},\bm{v}), \bm{x}) = w_0 + \sum_{i=1}^m x_i w_i +  \sum_{i=1}^m\sum_{j=i+1}^m x_ix_j  p(\bm{v}_i, \bm{v}_j, \bm{W}_{F(i), F(j)}) \
\end{equation}

where $\bm{W}_{F(i), F(j)}\in \mathbb{R}^{K \times K}$, and 

\begin{equation}
    p(\bm{v}_i, \bm{v}_j, \bm{W}_{F(i), F(j)}) = \sum_{k=1}^{K} \sum_{k'=1}^{K} v_{i}^{k} v_{j}^{k'} w_{F(i), F(j)}^{k, k'}
\end{equation}

OPNN also proposed to model the feature interactions via outer product. However, FmFM is different from OPNN in the following two aspects. First, FmFM is a simple shallow model without the fully connected layers as in~\cite{qu2016product}. We can use FmFM as a shallow component or a building block in any deep CTR models, like DeepFM~\cite{guo2017deepfm}. Second, we support field-specific variable embedding dimensions for features from different fields, which will be discussed in Section~\ref{sec:var_len}.

\subsection{FFM and FmFM, Memorization vs Inference}
\label{sec:ffm_vs_fmfm}
Unlike other factorization machines above, FFM cannot be reformed into the FmFM framework since it has a different way to look up their feature embeddings, we demonstrate its interaction terms' calculation in Figure \ref{fig:ffm}. FFM never share the feature embedding; it always looks up the field specific embeddings from the embedding tables. Thus, there are $n-1$ embeddings for one single feature, which are prepared to cross $n-1$ other fields respectively. Those field-specific embeddings will be learned independently during the training process, and there are no restrictions among those embeddings even belonging to the same feature.

\begin{figure}[h]
  \centering
  \includegraphics[width=0.8 \linewidth]{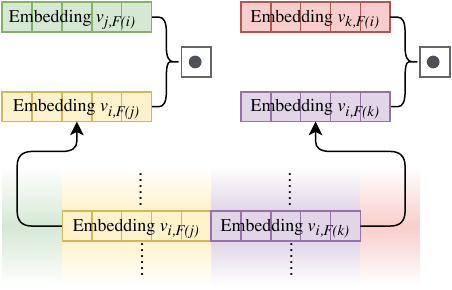}
  \caption{An example of FFM}
  \label{fig:ffm}
\end{figure}

This FFM mechanism gives the model maximal flexibility to fit the data, and the huge number of embedding parameters  also has incredible memorization capacity. Meanwhile, there is always a risk of over-fitting with it, even when there are billions of instances to be trained. The nature of the features' distribution is a long tail distribution, instead of a uniform distribution, that makes the feature pairs' distribution  even more imbalanced.

Given an example in Fig.\ref{fig:ffm}, assume that feature pair $(v_i, v_j)$ is a high frequency combination, while $(v_i, v_k)$ is a low frequency (possibly 0 frequency, or never appeared), since $v_{i, F(j)}$ and $v_{i, F(k)}$ are 2 independent embeddings in the setting of FFM, thus embedding $v_{i, F(j)}$ may be trained well but $v_{i, F(k)}$ may not. Due to the long tail distribution of features, those high frequent features pairs may dominate the number of training data, while other low frequency features which dominate the number of features, cannot be trained well.

FmFM use shared embedding vectors, as there is only one copy for each single feature. It utilizes a transformation process to project this single embedding vector into other $n-1$ fields. This is basically an inference process, and those $n-1$ vectors $v_{i,F(*)}$ are actually tied with the original embedding vector $v_i$. With those field matrices, the vectors are transformable forward and backward. That is the fundamental difference between FFM and FmFM; those transformable intermediate vectors within the same feature help the model learn those low frequency feature pairs well.

Back to the example in Fig.\ref{fig:fmfm}, even the feature pair $(v_i, v_k)$ is of low frequency, the feature embedding $v_i$  can still be trained well through the other  high frequency feature pairs like $(v_i, v_j)$, and the field matrix $M_{F(i),F(j)}$ can be trained well through other feature interactions between field $F(i)$ and field $F(k)$. Thus,  if the low frequency feature pair $(v_i, v_k)$ occurs during evaluation or test, the intermediate vector $v_{i,F(k)}$ can be inferred through $v_i M_{F(i),F(j)}$.

Despite this difference between FFM and FmFM, they have more in common. An interesting observation between figure \ref{fig:ffm} and figure \ref{fig:fmfm} is that, when all transformations are done, the FmFM model becomes a FFM model. We can cache those intermediate vectors and avoid matrix operations during runtime; the details will be discussed in the next section.

In contrast, FFM model cannot be reformed to a FmFM model, as we have mentioned above. Those $n-1$ field features embedding tables are independent, thus it is hard to compress them into one single feature embedding table and restore them when needed.

\subsection{Model Complexity}
\label{subsec:model_complexity}

The number of parameters in FM is $m + mK$, where $m$ accounts for the weights for each feature in the linear part $\{w_i|i=1,...,m\}$ and $mK$ accounts for the embedding vectors for all the features $\{\bm{v}_i|i=1,...,m\}$. FwFM use $n(n-1)/2$ additional parameters $\{r_{F_k, F_l} | k, l = 1,...,n\}$ for each field pair so that the total number of parameters of FwFM is $m + mK + n(n-1)/2$. The additional matrices in FmFM is $n(n-1)/2$ as compared to FM, and it has extra $\frac{n(n-1)}{2}K^2$ parameters.

For FFM, the number of parameters is $m + m(n-1)K$ since each feature has $n-1$ embedding vectors. Given that usually  $n\ll m$, the number of parameters  of other Factorization Machines are comparable with that of FM and significantly less than that of FFM. In Table~\ref{table:model_complexity} we compare the model complexity of all models mentioned so far. We also list the estimated number of parameters in the setting of section \ref{sec:exp} for each model, which use the public data set Criteo. Those numbers can give us an intuitive impression about the size of each model. FM, FwFM, and FmFM have similar sizes while FFM have more than dozen times than others.
 
\begin{table}[h]
\centering
	\begin{tabular}{| c | c | c |}
    \hline
    Model & N of Parameters & Estimated N in Criteo Dataset\\
    \hline
    LR & $m$ & 1.33M\\
    \hline
    Poly2 & $m + H$ & 45M\\
    \hline
    FM & $m + mK$ & 14.63M\\
    \hline
    FwFM & $m + mK + \frac{n(n-1)}{2}$ & 14.63M\\
    \hline
    FmFM & $m + mK + \frac{n(n-1)}{2}K^2$ & 14.63M\\
    \hline
    FFM & $ m + m(n-1)K$ & 859.18M\\
    \hline
	\end{tabular}
    \caption{A summary of model complexities (ignoring the bias term $w_0$). The estimate of the total $N$ of the model in the settings of Section \ref{subsec:performance_comparison} }
    \label{table:model_complexity}
\end{table}

\section{Model Optimization}
\label{sec:model_opt}
In this section we present our methodologies to optimize the FmFM model. There are 3 tactics which we can devise to reduce the complexity of FmFM further. In section \ref{sec:var_len} we introduce the field-specific embedding dimensions, which is a unique property in FmFM; it allows us to use field specific dimensions in the embedding table, instead of a fixed length $K$ globally. In Section \ref{sec:vec_cache}  we introduce the method to cache the intermediate vectors to avoid the matrix operations, which can reduce the FmFM model's computational complexity in runtime. In Section \ref{sec:linear_terms} we introduce the method to reduce the linear terms and replace them with field specific weights.

\subsection{Field-specific Embedding Dimensions}
\label{sec:var_len}
The main improvement of FM over LR  model is that, FM use the embedding vector to represent each feature. In order to make the dot product, it requires the vector dimension $K$ of all feature embeddings to be the same, even though features come from different fields. The improved models like FwFM, FvFM also adopt this property. The vector dimension matters both in model complexity and model performance, the work \cite{pan2018fwfm} discussed this trade-off between performance and time cost, but the vector dimension can only be optimized globally.

When we utilize the matrix multiplication in FmFM, it actually does not require the field matrices to be square matrices;  we can adjust the output vector length by changing the shape of the field matrix. This property gives us an another flexibility to set the field-specific lengths on-demand in the embedding table, as we show in figure \ref{fig:var_emb}. 

\begin{figure}[h]
  \centering
  \includegraphics[width=0.8\linewidth]{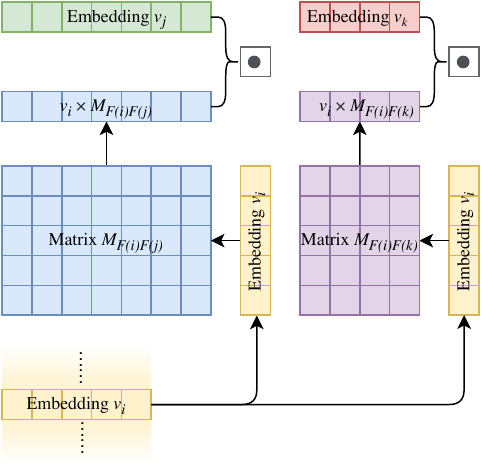}
  \caption{An Example of Variable Vector Length of Embeddings}
  \label{fig:var_emb}
\end{figure}

The dimension of an embedding vector determines the amount of information it can carry; this property allows us to accommodate the need for each field. For the example $(i,j)$ in Fig.\ref{fig:var_emb}, the field $user\_gender$ may contain only 3 values (male, female, other), while another field $top\_level\_domain$ may contain more than 1 million features. Thus, the embedding table of field $user\_gender$ may only need  5-dimension (5D), while the field $top\_level\_domain$ may need 7D. The field matrix $M$ should be set up with a shape in $(7,5)$. When we cross the feature between $top\_level\_domain$  and $user\_gender$, the matrix can transfer the 5D feature vector to a 7D vector, making it ready to do a dot product with the feature vector from field $top\_level\_domain$.

To optimize the field-specific embedding vector dimension without model performance loss, we propose a 2-pass method. In the first pass, we use a larger fixed embedding vector dimension for all fields, e.g. $K=16$, and train the FmFM as a full model. From the full model, we learn how much information (variance) in each field, then we utilize a standard PCA dimensionality reduction to the embedding table in each field. From the experiment in Section \ref{subsec:dim_opt} we found that the new dimension which contains 95\% original variance is a good trade-off. With this new field specific dimension setting, we train the model in a second pass from scratch, and should get the second model without any significant performance loss, compared to the first full model. 

\begin{table}[h]
\begin{tabular}{ccr}
\toprule
\makecell{Feature\\Field ID} & \makecell{Emb\\Dim} & \makecell{Feat. N\\ in Field}\\
\midrule
Field \#01 & 3 & 62 \\
Field \#02 & 8 & 113 \\
Field \#03 & 5 & 125 \\
Field \#04 & 7 & 50\\
Field \#05 & 9 & 223\\
Field \#06 & 8 & 147\\
Field \#07 & 6 & 99\\
Field \#08 & 5 & 78\\
Field \#09 & 8 & 103\\
Field \#10 & 3 & 8\\
Field \#11 & 5 & 31\\
Field \#12 & 3 & 56\\
Field \#13 & 6 & 81\\
Field \#14 & 8 & 1,457\\
Field \#15 & 12 & 555\\
Field \#16 & 2 & 245,195\\
Field \#17 & 11 & 166,164\\
Field \#18 & 5 & 305\\
Field \#19 & 4 & 18\\
Field \#20 & 14 & 12,054\\
\bottomrule
\end{tabular}
\begin{tabular}{ccr}
\toprule
\makecell{Feature\\Field ID} & \makecell{Emb\\Dim} & \makecell{Feat. N\\ in Field}\\
\midrule
Field \#21 & 8 & 633\\
Field \#22 & 2 & 3\\
Field \#23 & 13 & 46,329\\
Field \#24 & 14 & 5,228\\
Field \#25 & 8 & 243,452\\
Field \#26 & 13 & 3,176\\
Field \#27 & 4 & 26\\
Field \#28 & 14 & 11,744\\
Field \#29 & 10 & 225,320\\
Field \#30 & 6 & 10\\
Field \#31 & 14 & 4,726\\
Field \#32 & 12 & 2,056\\
Field \#33 & 2 & 3\\
Field \#34 & 9 & 238,638\\
Field \#35 & 4 & 16\\
Field \#36 & 6 & 15\\
Field \#37 & 12 & 67,854\\
Field \#38 & 7 & 87\\
Field \#39 & 11 & 50,940\\
\\
\bottomrule
\end{tabular}
\caption{The optimized dimensions for each field in Criteo data set, 95\% variance kept}
\label{table:field_dim}
\end{table}

Table \ref{table:field_dim} shows the optimized dimension for each field in the Criteo dataset, with the PCA method. This list shows that the range of those dimensions are huge which from 2 to 14, and most of the dimensions are less than 10. The average $\overline{K}$ is only 7.72, which is less than the optimal setting in the FwFM. With keeping most variance from the dataset, the lower average dimension means the model has fewer parameters, requires less memory.

\subsection{Intermediate Vectors Cache}
\label{sec:vec_cache}
FmFM is a lower complexity model than FFM, in the number of parameters; however it requires  expensive matrix operations in the transformation step. In table \ref{table:model_flop}, we list the number of Floating Point Operations (FLOPs) for each model\footnote{We use the following values to estimate the FLOPs: $n=39, K=16$, $H=200$ denotes the number of nodes in the hidden layers of DNN, $L=3$ denotes the number of hidden layers in DNN, $K'=32$ denotes the dimension of embedding vectors in the new space in AutoInt, and $s_{\text{FwFM}}=90\%, s_{\text{DNN}}=80\%$ denotes the sparsity rate of the FwFM and DNN component in DeepLight.}, and estimate it with typical settings.

\begin{table}[h]
\centering
	\begin{tabular}{| c | c | r |}
    \hline
    Model &  Floating Point Operations  & Estimated \#\\
    \hline
    LR & $O(n)$ & 78\\
    \hline
    Poly2 & $O(n^2)$ & 1,560\\
    \hline
    FM & $O(nK)$ & 1,920\\
    \hline
    FwFM & $O(n^2K)$ & 25,272\\
    \hline
    FFM & $O(n^2K)$ & 24,531\\
    \hline
    FmFM & $O(n^2K^2)$ & 403,923\\
    \hline
    FmFM(cached) & $O(n^2K)$ & 24,531\\
    \hline
    \textbf{\makecell{FmFM(cached\\ \& 95\% variance)}} & $O(n^2\overline{K})$ & \textbf{8,960}\\
    \hline
    \hline
    Wide \& Deep & $O(n^2 + nKH + LH^2)$ & \char`\~ 500,000 \\
    \hline
    Deep \& Cross & $O(n^2K^2 + nKH + LH^2)$ & \char`\~ 510,000 \\
    \hline
    DeepFM  & $O(nKH + LH^2 + n^2)$ & \char`\~ 246,000 \\
    \hline
    xDeepFM & $O(nH^2KL)$ & \char`\~ 150,000,000 \\
    \hline
    AutoInt & $O(nHK'(n+K))$ & \char`\~ 28,000,000 \\
    \hline
    FiBiNET & $O(n^2K^2 + n^2KH + LH^2)$ & \char`\~ 10,000,000 \\
    \hline
    DeepLight & \makecell{$O(n^2K (1-s_{\text{FwFM}}) +$ \\ $(nKH + LH^2) (1-s_{\text{DNN}}))$} & \char`\~ 102,000 \\
    \hline
	\end{tabular}
    \caption{A summary of Floating Point Operations by model}
    \label{table:model_flop}
\end{table}


Among those Factorization Machine models, the FmFM needs the most operations to accomplish its calculation, which is about $k$ times as FwFM and FFM, but still faster than most DNN models. In section \ref{sec:ffm_vs_fmfm}, we have shown that a FmFM model can be transformed into a FFM model, by caching all intermediate vectors. In this sense, we can reduce its number of operations to the same magnitude as FM and FFM, which is almost 20 times faster.


\subsection{Embedding Dimension and Cache Optimization Combined}
\label{sec:opt_combined}
When we combine the field-specific embedding dimensional optimization and the cache optimization, the inference speed can be much faster, and the required memory can be reduced significantly. This benefits from another property of FmFM - the interaction matrices are symmetrical, which means
\begin{equation}\label{equ:matrix_sym}
\langle \bm{v}_iM_{F(i)F(j)}, \bm{v}_j\rangle = \langle \bm{v}_jM^T_{F(i)F(j)}, \bm{v}_i\rangle
\end{equation}
We have a proof for this lemma in the Appendix.

Thus, we can choose to cache those intermediate vectors which have lower field dimensions, and dot-product with the other feature vectors during inference. For example, in the setting of table \ref{table:field_dim}, two features $v_i$ and $v_j$ are from field \#16 and \#28 respectively. With this property, when we calculate the interaction between $v_i$ and $v_j$, we can cache either $\bm{v}_iM_{16,28}$ or $\bm{v}_jM^T_{16,28}$.

Since the field matrix $M_{16,28}$ has a shape of $[2,14]$, the former one $\bm{v}_iM_{16,28}$ increased the dimension from 2 (field \#16) to 14 (intermediate vector), then dot-product with $\bm{v}_j$ whose dimension is also 14. It costs 14 units of memory for the intermediate vectors cache, and takes $2\times 14$ FLOPs during inference. By contrast, the latter one $\bm{v}_jM^T_{16,28}$ reduces the dimension from 14 (field \#28) to 2 (intermediate vector), then dot-product with $\bm{v}_i$ whose dimension is also 2. It costs 2 units of memory for the intermediate vectors cache, and takes $2\times 2$ FLOPs during inference. In this single field pair, the optimized cache with field-specific embedding dimensions can save 7 times memory and time  without any precision loss. 

With those two optimization techniques combined, the FmFM model's time complexity is reduced drastically; in table \ref{table:model_flop}, we estimate that the optimized model only takes 8,960 FLOPs, which is only about 1/3 of FFM. In the section\ref{subsec:dim_opt}, we will show that this optimized model can achieve the same performance as the full model.

\subsection{Soft Pruning}
The field-specific embedding dimensions also act in a role similar to pruning actually; while traditional pruning such as DeepLight~\cite{deng2020sparse} gives a binary decision to keep or drop a field or a field pair, the field-specific embedding dimensions give us a new way to determine the importance of each field and field pair on-demand, and assign each field a factor to represent its importance. For example, in the FmFM model of Table\ref{table:field_dim}, the cross field \#17 and \#20 is a high strength pair; it takes 11 units of cache and $2\times11$ FLOPs during inference; in contrast, a low strength pair, field \#18 and \#22, only takes 2 units of cache and $2\times2$ FLOPs. 

When we drop a field pair in the traditional pruning, its signal was lost totally; while in this method, a field pair still keeps the major information with minimal cost. It is a soft version of pruning, which is similar to Softmax. It is more efficient and sees less performance drop during soft pruning.

\begin{figure}[h]
  \centering
  \includegraphics[width=\linewidth]{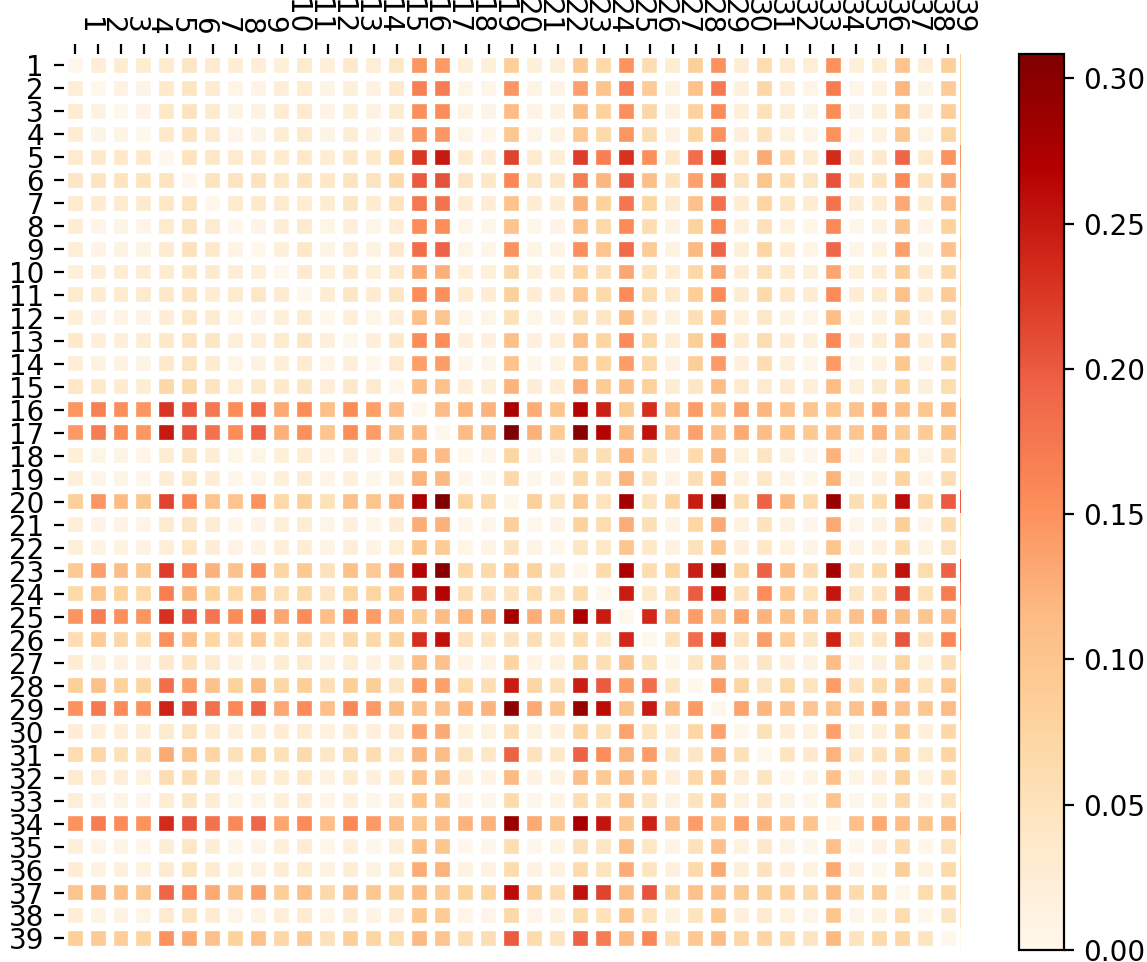}
  \caption{An example of Mutual Information Score between field pairs and label in Criteo Dataset}
  \label{fig:criteo_mi}
\end{figure}

Figure \ref{fig:criteo_mi} shows a heat-map of mutual information scores between field pairs and labels in the Criteo dataset, which represents the strength of field pairs in prediction. Figure \ref{fig:cross_dim} shows the cross field dimensions, which is the lower dimension between two fields (explanation in Section \ref{sec:opt_combined}); it represents the parameters and computational cost for each field pair. Obviously, these two heat-maps are highly related to each other, which means the optimized FmFM model allocates more parameters and more computation on those higher strength field pairs, and fewer parameters and less computation on lower strength field pairs.

\begin{figure}[h]
  \centering
  \includegraphics[width=\linewidth]{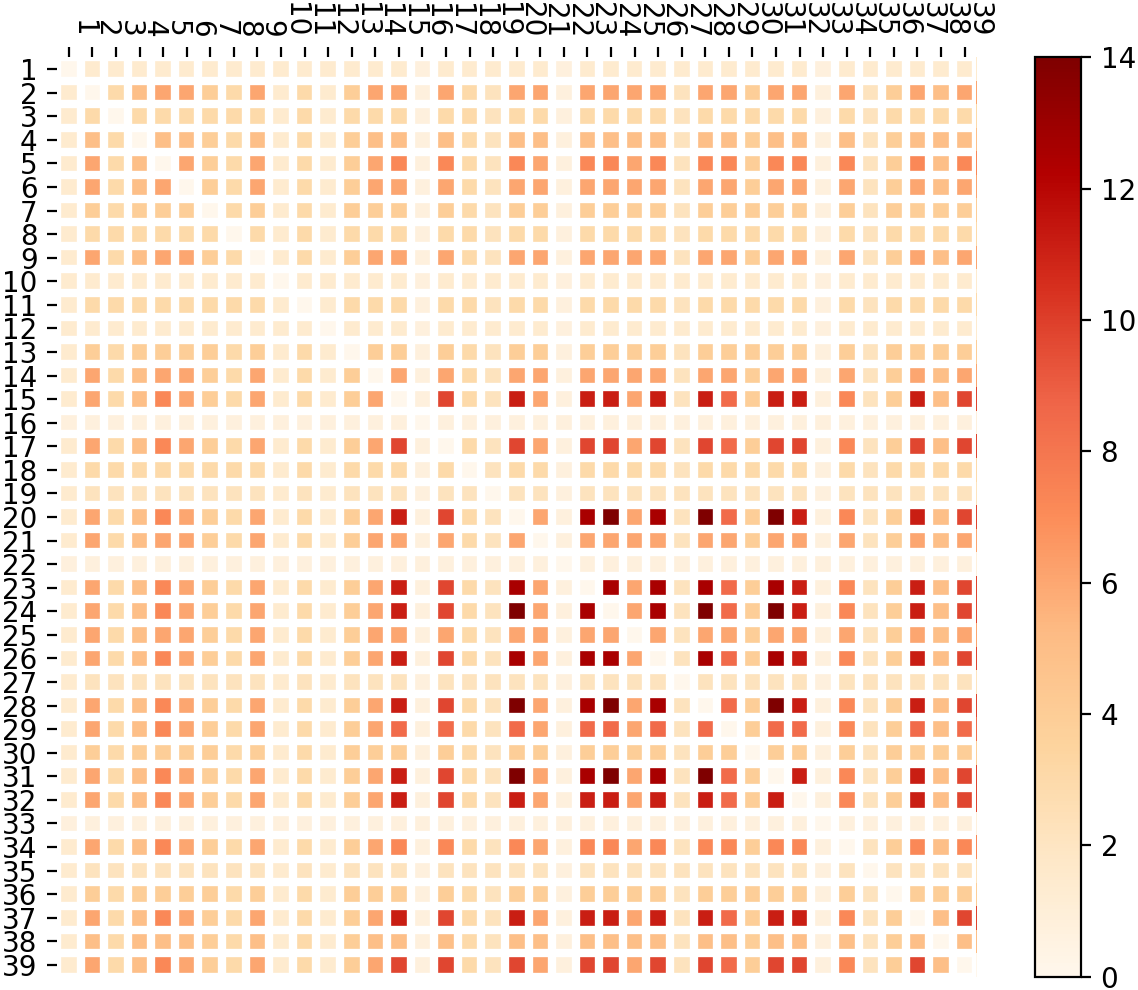}
  \caption{An example of cross fields dimensions - $min(D_i, D_j)$ in Criteo Dataset}
  \label{fig:cross_dim}
\end{figure}

\subsection{Linear Terms}
\label{sec:linear_terms}
There is a linear part in Eq.\ref{eq:fmfm}:

\begin{equation}
  \sum_{i=1}^m x_i w_i  
\end{equation}

which requires an extra scalar $w_i$ to be learned for each feature. However the learned embedding vector $v_i$ should contain more information, and the weight $w_i$ can be learned from $v_i$ by a simple dot product. Another benefit from the learned $v_i$ is that, it can help to learn the embedding vector from the linear part.

We follow the method from the work of \cite{pan2018fwfm}, by learning a field specific vector $w_{F(i)}$ so that all features from the same field $F(i)$ will share the same linear weight vector. Then the linear terms can be rewritten to:

\begin{equation}
    \sum_{i=1}^m x_i \langle \bm{v}_i,w_{F(i)} \rangle
\end{equation}

We apply this linear term optimization to FwFM, FvFM and FmFM by default in the rest of the paper.

\label{sec:cross_prun}

\section{Experiments}
\label{sec:exp}
In this section we present our experimental evaluation results. We will first describe the data sets and implementation details in Section~\ref{subsec:datasets} and ~\ref{subsec:implementation} respectively. In Section~\ref{subsec:performance_comparison} we compare FmFM with other baseline models like LR, FM, FwFM and FFM, as well as the state-of-the-art methods like Wide \& Deep, Deep \& Cross network, xDeepFM, AutoInt, FiBiNET and DeepLight. In Section \ref{subsec:dim_opt}, we did a few experiment on the Criteo dataset and observe the model performance change when we apply the field-specific dimension in embedding.

\subsection{Data sets}
\label{subsec:datasets}
We use 2 public data sets to evaluate our model performance: 

\begin{enumerate}
\item The first one is the Criteo CTR data set; it is a well-know benchmark data set which used for the Criteo Display Advertising Challenge~\cite{criteo-display-ad-challenge}. The Criteo data set is already label balanced, the positive to the negative ratio is about 1:3. There are 45 million samples and each sample has 13 numerical fields and 26 categorical fields.
\item The second data set is the Avazu CTR data set; it was used in the Avazu CTR prediction competition, which predicts whether a mobile ad will be clicked. The positive to the negative ratio in the Avazu data set is about 1:5. There are 40 million samples and each sample has 23 categorical fields.
\end{enumerate}

We follow those existing works ~\cite{guo2017deepfm,xdeepfm,shan2016deep,wang2017deep,zhang2016deep,zhou2018deep, deng2020sparse,autoint}, for each data set, we split it into 3 parts randomly, 80\% for training, 10\% for validation, and 10\% for testing. \footnote{Some of those works split 90\%:10\% for just training and testing, but we adopt the strict one with validation set and less training set.} 

Regarding to those numerical features in the Criteo data set, we adopt the log transformation of $log(x)^2$ if $x>2$, which proposed by the winner of the Criteo competition\footnote{https://www.csie.ntu.edu.tw/~r01922136/kaggle-2014-criteo.pdf} to normalize the numerical features. This method was also used by ~\cite{deng2020sparse} and ~\cite{autoint}. Regarding the date/hour feature in the Avazu data set, we transfer it into 2 features: day\_of\_week(0-6) and hours(0-23) to consume the feature better.

We also remove those infrequent features that are less than a threshold in both data sets and replace their values with the default "unknown" feature in that field. We set the threshold to 8 for the Criteo data set, and to 5 for the Avazu data set.

The statistics of the normalized data sets are shown in Table~\ref{table:datasets}.

\begin{table}[h]
\centering
    \begin{tabular}{ | c | c | r | c | c | c |} 
    \hline
    \multicolumn{2}{|c|}{Data set} & Samples & Fields & Features &  Pos:Neg \\
    \hline
    \multirow{3}{*}{Criteo} & Train & 36,672,493 & \multirow{3}{*}{39} & \multirow{3}{*}{1,327,180} & \multirow{3}{*}{\char`\~ 1:3}\\
    \cline{2-3}
    & Validation & 4,584,062 &  &  &\\
    \cline{2-3}
    & Test & 4,584,062 &  &  &\\
    \hline
    \multirow{3}{*}{Avazu} & Train & 32,343,173 & \multirow{3}{*}{23} & \multirow{3}{*}{1,544,257} & \multirow{3}{*}{\char`\~ 1:5}\\
    \cline{2-3}
    & Validation & 4,042,897 &  & &\\
    \cline{2-3}
    & Test & 4,042,897 &  & &\\
    \hline
    \end{tabular}
    \caption{Statistics of training, validation and test sets of the Criteo data sets.}
    \label{table:datasets}
\end{table}

\subsection{Experiment Setup}
\label{subsec:implementation}

We have implemented the LR (logistic regression) and all factorization machine models (FM, FwFM, FFM, FvFM and FmFM) with Tensorflow.\footnote{We open-sourced the training code and the feature extraction code at \\ https://github.com/VerizonMedia/FmFM}. We follow the implementation of LR and FM in~\cite{qu2016product}, and implement FFM, FwFM following the work \cite{pan2018fwfm}. 

We evaluate all models performance by AUC (Area Under the ROC Curve) and Log Loss on the test set. It is noticeable that a slightly higher AUC or lower Log Loss at 0.001-level is regarded a significant improvement for CTR prediction task, which has also been pointed out in existing works \cite{autoint, cheng2016wide, wang2017deep}.

For those state-of-the-art models, they are all DNN models and may need more hyper-parameters tuning, we pull their performance (AUC and Log Loss) from their original papers, in order to keep their results optimal. It is fair to compare our results with theirs, since we use more strict data splittings; while their implementations may have slight differences, e.g. feature processing, optimizer (Adam or Adagrad), we list their results just for reference. The Deep \& Cross Network is an exception, since their paper only listed the Log Loss but not AUC. Thus, we implemented their model and got a similar performance. 

\subsection{Performance Comparisons}
\label{subsec:performance_comparison}

In this section we will conduct performance evaluation for FmFM. We will compare it with LR, FM, FwFM and FFM on the two data sets mentioned above. We always use the full size model to compare in the results; that means we don't use any optimization methods mentioned in Section \ref{sec:model_opt}. For the parameters such as regularization coefficient $\lambda$, and learning rate $\eta$ in all models, we select those which lead to the best performance on the validation set and then use them in the evaluation on the test set. Experiment results can be found in Table~\ref{table:banchmark_criteo} for the Criteo data set, and Table~\ref{table:banchmark_avazu} for the Avazu data set. 

\begin{table}[h]
\begin{tabular}{| c | c | c | c | c |} 
    \hline
	\multirow{2}{*}{Models}  & \multicolumn{3}{c|}{AUC} & \multirow{2}{*}{\makecell{Log Loss \\ (Test Set)} }\\
	\cline{2-4}
    & Training & Validation & Test &  \\
	\hline
	LR & 0.7930 & 0.7918 & 0.7917 & 0.4582\\
    \hline
    \hline
	FM & 0.8142 & 0.8075 & 0.8075 & 0.4441\\
	\hline
	FFM & \textbf{0.8230} & 0.8103 & 0.8103 & 0.4414\\
	\hline
	FwFM & 0.8191 & 0.8092 & 0.8092 & 0.4426\\
    \hline
	FvFM(ours) & 0.8192 & 0.8102 & 0.8101 & 0.4417\\
    \hline
	FmFM(ours) & 0.8183 & \textbf{0.8109} & \textbf{0.8109} & \textbf{0.4410}\\
	\hline
	\hline
    Deep \& Cross & 0.8244 & 0.8118 & 0.8118 & 0.4413\\
    \hline
    Wide \& Deep & - & - & 0.7981 & 0.4677\\
    \hline
    DeepFM  & - & - & 0.8007 & 0.4508 \\
    \hline
    xDeepFM & - & - & 0.8052 & 0.4418 \\
    \hline
    AutoInt & - & - & 0.8061 & 0.4454 \\
    \hline
    FiBiNET & - & - & 0.8103 & 0.4423 \\
    \hline
    DeepLight & - & - & \textbf{0.8123} & \textbf{0.4395} \\
    \hline
    \end{tabular}
\caption{Comparison among models on the Criteo CTR data sets.}
\label{table:banchmark_criteo}
\end{table}

\begin{table}[h]
\begin{tabular}{| c | c | c | c | c |} 
    \hline
	\multirow{2}{*}{Models}  & \multicolumn{3}{c|}{AUC} & \multirow{2}{*}{\makecell{Log Loss \\ (Test Set)} }\\
	\cline{2-4}
    & Training & Validation & Test &  \\
	\hline
	LR & 0.7526 & 0.7521 & 0.7517 & 0.3953\\
    \hline
    \hline
	FM & 0.7744 & 0.7696 & 0.7695 & 0.3857\\
	\hline
	FFM & \textbf{0.8012} & 0.7761 & 0.7761 & 0.3826\\
	\hline
	FwFM & 0.7822 & 0.7730 & 0.7731 & 0.3835\\
    \hline
	FvFM(ours) & 0.7836 & 0.7732 & 0.7733 & 0.3834\\
    \hline
	FmFM(ours) & 0.7943 & \textbf{0.7764} & \textbf{0.7763} & \textbf{0.3822}\\
	\hline
	\hline
    Deep \& Cross & 0.8109 & 0.7825 & 0.7826 & 0.3791\\
    \hline
    AutoInt & - & - & 0.7752 & 0.3823 \\
    \hline
    Fi-GNN & - & - & 0.7762 & 0.3825 \\
    \hline
    FGCNN+IPNN & - & - & 0.7883 & 0.3746 \\
    \hline
    DeepLight & - & - & \textbf{0.7897} & \textbf{0.3748} \\
    \hline
    \end{tabular}
\caption{Comparison among models on Avazu CTR data sets.}
\label{table:banchmark_avazu}
\end{table}

We observe that FvFM and FmFM can achieve better performance than LR, FM, and FwFM on both data sets, which is in our expectation. Surprisingly, the FmFM can achieve better performance than FFM constantly in both test sets. As we mentioned before, even though FFM is a model dozens times larger than FmFM, our FmFM model still get the best AUC in the test set among all shallow models. Additionally, if we compare the differences in AUC between training and test, we found that the $\Delta AUC_{FmFM}=0.0074$ is the lowest one among those factorization machine models, which affirms our hypothesis in section \ref{sec:ffm_vs_fmfm}: those low frequency features are also trained well with the help of the interaction matrix, and this mechanism helps FmFM to avoid over-fitting.

\subsection{Embedding Dimension Optimization}
\label{subsec:dim_opt}
In this part we implement the method described in \ref{sec:var_len}, whereby we have a full size model, we can extract its embedding tables for each field, then we utilize a standard PCA dimensionality reduction. Here we do several experiments and compare how the dimensionality reduction impact the model performance, and try to find a trade-off between the model size, speed and its performance.

We use the full size FmFM model from experiment \ref{subsec:performance_comparison} on the Criteo data set to evaluate our metrics. We keep 99\%, 97\%, 95\%, 90\%, 85\% and 80\% variance in PCA dimensionality reduction respectively, and estimate the average embedding dimensions and float operations (FLOPs, with cached intermediate vectors). With the new dimensions setting, we train those FmFM models the second pass respectively, and observe the AUC and Log Loss change in test set.

Table \ref{table:model_dim_opt} shows the summary of experiments. The average embedding dimension was reduced significantly when we keep less variance from PCA: there is only less than 1/2 embedding dimensions and 1/3 computation cost when we keep 95\% variance, while there is no significant change on the model's performance compare to the full size model. Thus, 95\% variance is a good trade-off when we optimize the embedding dimensions in FmFM.

\begin{table}[h]
\centering
	\begin{tabular}{| c | r | r | r | r |}
    \hline
    \makecell{Variance\\ \%} &  \makecell{Emb Dim\\(Average)} & \makecell{FLOPs\\Estimated \#} & \makecell{AUC\\(Test Set)} & \makecell{Log Loss\\(Test Set)}\\
    \hline
    Full & 16(100\%) & 24,531(100\%) & 0.8109 & 0.4410 \\
    \hline
    99\% & 10.56(66.0\%) & 12,884(52.5\%) & 0.8109 & 0.4410 \\
    \hline
    97\% & 8.69(54.3\%) & 10,280(41.9\%) & 0.8107 & 0.4411 \\
    \hline
    \textbf{95\%} & \textbf{7.72(48.2\%)} & \textbf{8,960(36.5\%)} & \textbf{0.8108} & \textbf{0.4411} \\
    \hline
    90\% & 6.26(39.1\%) & 7,202(29.4\%) & 0.8103 & 0.4415 \\
    \hline
    85\% & 3.82(23.9\%) & 4,716(19.2\%) & 0.8084 & 0.4432 \\
    \hline
    80\% & 3.36(21.0\%) & 4,392(17.9\%) & 0.8080 & 0.4436 \\
    \hline
	\end{tabular}
    \caption{Compare among FmFM optimized models with embedding dim optimization, an example of the Criteo Data Set}
    \label{table:model_dim_opt}
\end{table}

Figure \ref{fig:auc_flop} shows these models' performance (in AUC) and their computational complexity (in FLOPs). As a shallow model, the optimized FmFM model gets higher AUC as well as lower FLOPs,  compared with all the baseline models except Deep \& Cross and DeepLight. While its computational cost is much lower than these two complex models which ensembled DNN module and shallow module, its FLOPs is only 1.76\% and 8.78\% of them, respectively. The lower FLOPs makes it preferable when the  computation latency is strictly limited, which is the common scenario in the real-time online ads CTR prediction and recommender systems.

\begin{figure}[h]
  \centering
  \includegraphics[width=\linewidth]{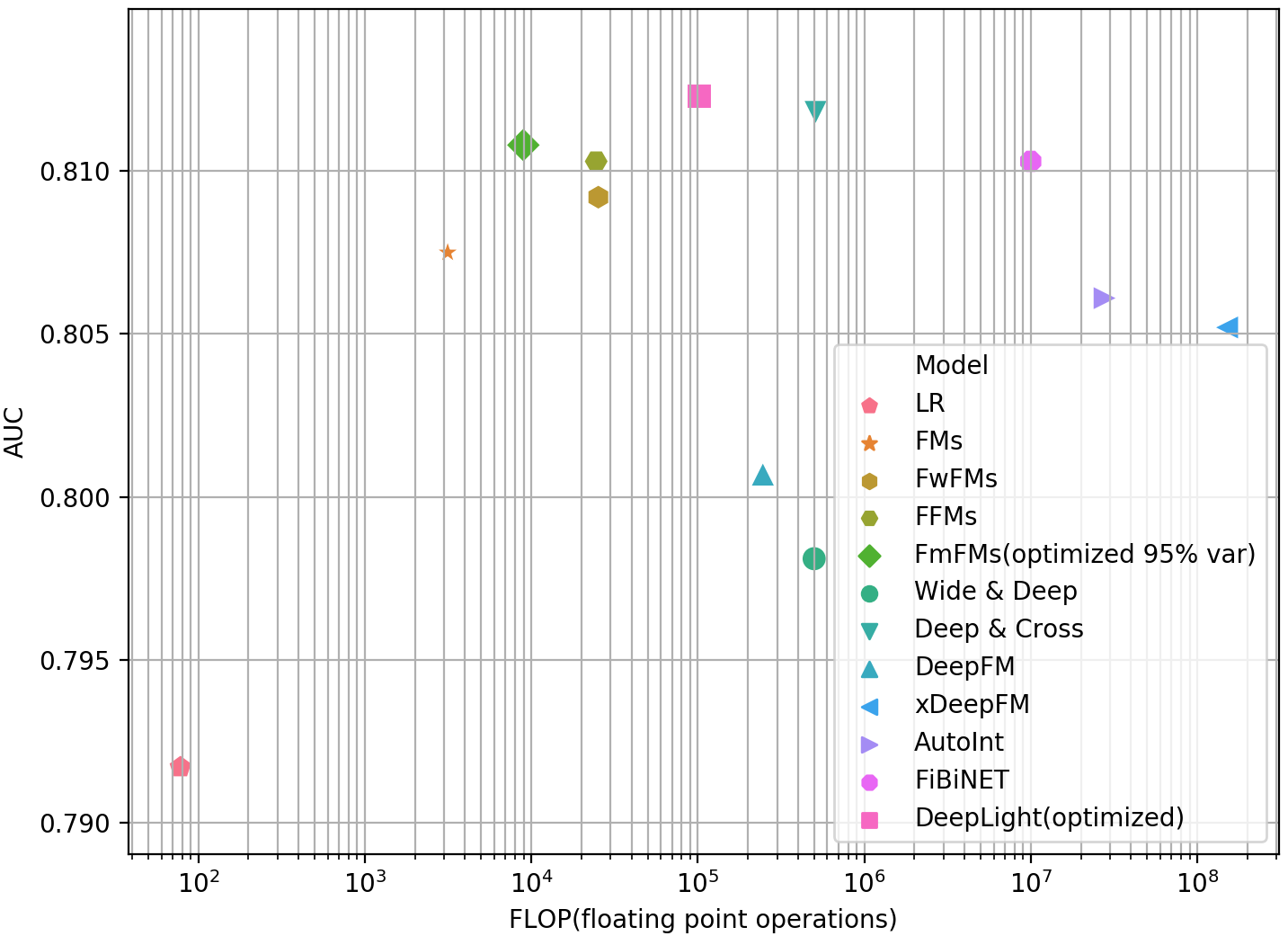}
  \caption{AUC and FLOPs comparison among all models on the Criteo dataset}
  \label{fig:auc_flop}
\end{figure}

\label{subsec:cross_prun}

\section{Conclusion and Further Works}
In conclusion, we propose a novel approach FmFM to model the interactions of field pairs as a matrix. We prove that FmFM is a unified framework of factorization machine model family, in which both FM and FwFM can be treated as special cases. We devise a few optimizations to FmFM, including field-specific embedding dimensions and caching intermediate vectors. These optimizations make the FmFM lightweight and faster during inference, taking only thousands of floating-point operations to make a prediction. We have done comprehensive experiments to verify the effectiveness and efficiency of the proposed model. It achieves state-of-the-art performance among all shallow models, including FM, FFM and FwFM, and its performance is even comparable to those complex DNN models. 

With regard to future work, there are a few potential research directions:
\begin{itemize}
    \item The FmFM is still a linear model, since the field interaction are matrices, and embedding vectors are transformed linearly. We can introduce the non-linear layers to the field interaction and let the model become a non-linear model, which is more flexible.
    \item All the factorization machine models are actually Degree-2 models, which allows up to 2 fields interactions. This restriction is majorly because the dot product. In the future, we can introduce the 3D tensor and allows the 3 fields interaction, or even higher ranks. This work may require more model optimization since there are too much Degree-3 interactions.
    \item We can combine the DNN models like the Wide and Deep~\cite{cheng2016wide}, DeepFM~\cite{guo2017deepfm}, DeepLight~\cite{deng2020sparse}, and try FmFM as a building block in DNN models to further improve their performances. We believe this method should be more competitive in the model performance, when compare to those deep learning based models in Section \ref{sec:related_works}.
\end{itemize}


\bibliographystyle{ACM-Reference-Format}
\bibliography{reference}

\appendix
\section{Math Proof}
\begin{lemma}
Given two row vectors $v_i$ and $v_j$ whose lengths are $k$ and $l$ respectively, there is a matrix $M\in \mathbb{R}^{k,l}$, then:

\begin{equation}
\bm{v}_i\times M\cdot \bm{v}_j = \bm{v}_j\times M^T\cdot \bm{v}_i
\end{equation}

where $\times$ denotes matrix multiplication, and $\cdot$ denotes dot product.
\end{lemma}

\begin{proof}
Since $\bm{v}_i\times M\cdot \bm{v}_j$ is a scalar, we denote it as $a$, and the dot product can be rewrite to a matrix multiplication, we rewrite the left of the equation:
\begin{equation}
    a = \bm{v}_i\times M\cdot \bm{v}_j = \bm{v}_i\times M \times \bm{v}^T_j
\end{equation}

while the transpose of a scalar equals to itself:
\begin{equation}
    a^T = (\bm{v}_i\times M \times \bm{v}^T_j)^T = \bm{v}_j\times M^T\times \bm{v}^T_i =\bm{v}_j\times M^T\cdot \bm{v}_i= a
\end{equation}

Hence, the left equals the right.
\end{proof}

\end{document}